 \documentclass[twoside,reqno,11pt]{fcaa} 

\usepackage{graphicx}
\usepackage{epsfig}
\usepackage{amsthm}
\usepackage{amsmath}
\usepackage{latexsym}
\usepackage{amsfonts}
\usepackage{amssymb}
\usepackage{amstext}
\usepackage{textgreek}
\usepackage{mathtools}
\usepackage{xcolor}
\usepackage{mathrsfs} 
\usepackage{multirow}
\usepackage{mathabx}
\usepackage[caption=false]{subfig}

\newtheoremstyle{theorem}
  {15pt}          
  {15pt}  
  {\sl}  
  {\parindent}
  {\sc}  
  {. }   
  { }    
  {}     

  \theoremstyle{theorem}
\newtheorem{example}{Example}[section]
\newtheorem{theorem}{Theorem}[section]
\newtheorem{corollary}[theorem]{Corollary}
\newtheorem{lemma}[theorem]{Lemma}

\numberwithin{figure}{section}

 \def\proofname{\indent\rm P\ r\ o\ o\ f.}
 
 \def\qed{\hfill$\Box$} 

\DeclareRobustCommand{\greektext}{%
  \fontencoding{LGR}\selectfont\def\encodingdefault{LGR}}
\DeclareRobustCommand{\textgreek}[1]{\leavevmode{\greektext #1}}
\ProvideTextCommand{\~}{LGR}[1]{\char126#1}

 \usepackage{hyperref} 

 \def\theequation{\arabic{section}.\arabic{equation}}

 \def\themycopyrightfootnote{\vspace*{3pt}
 \copyright \, Year\,  Diogenes  Co., Sofia
   \par  \noindent pp. xxx--xxx
   \hfill  \vspace*{-36pt}
  }

 \title[Fractional SI Model]{Fractional-order susceptible-infected model: definition and applications to the study of COVID-19 main
protease}
\author[Abadias, Estrada-Rodriguez, Estrada]{Luciano Abadias$^{1,2}$, Gissell Estrada-Rodriguez$^{3}$, Ernesto
Estrada$^{2,4}$}

\begin{document}

 \vbox to 2.5cm { \vfill }


 \bigskip \medskip

 \begin{abstract}

We propose a model for the transmission of perturbations across the
amino acids of a protein represented as an interaction network. The
dynamics consists of a Susceptible-Infected (SI) model based on the
Caputo fractional-order derivative. We find an upper bound to the
analytical solution of this model which represents the worse-case
scenario on the propagation of perturbations across a protein residue
network. This upper bound is expressed in terms of Mittag-Leffler
functions of the adjacency matrix of the network of inter-amino acids
interactions. We then apply this model to the analysis of the propagation
of perturbations produced by inhibitors of the main protease of SARS
CoV-2. We find that the perturbations produced by strong inhibitors
of the protease are propagated far away from the binding site, confirming
the long-range nature of intra-protein communication. On the contrary,
the weakest inhibitors only transmit their perturbations across a close
environment around the binding site. These findings may help to the
design of drug candidates against this new coronavirus.

 \medskip

{\it MSC 2010\/}: Primary 26A33; Secondary 33E12, 92C40, 05C82

 \smallskip

{\it Key Words and Phrases}: Caputo derivative; Mittag-Leffler matrix functions; Susceptible-Infected
model; COVID-19, SARS CoV-2 protease

 \end{abstract}

 \maketitle

 \vspace*{-16pt}



\section{Introduction}

The presence of a networked structure is one of the fundamental characteristics
of complex systems in general \cite{Networks_1,Networks_2}. It could
be argued that the main function of such networks is that of allowing
the communication between the entities that form its structure. In
the case of proteins, the non-covalent interactions between residues
in their three-dimensional structures form inter-residue networks
\cite{Networks_1,PRN_2}. These networks facilitate that information
about one site is transmitted to and influences the behavior of another.
This phenomenon--the transmission of any perturbation in protein
structure and function from one site to another--is known as allostery,
which represents an essential feature of protein regulation and function
\cite{Communication_1,Communication_2}. Allostery permits that two
residues geometrically distant can interact with each other. As observed
experimentally by Ottemann et al. \cite{Ottemann} a conformational
change of 1{\AA} in a residue can be transmitted to another 100{\AA} apart.
As stated long time ago, such allosteric effects can occur even when
the average protein structure remains unaltered \cite{Cooper_Dryden}.
An important kind of allosteric effect is the transmission of the
changes produced by a ligand interacting with a protein. Such transmission
occurs from the residues proximal to the binding site to other residues
distant from it. Such kind of allosteric interaction is very important
for understanding the effects of drugs on their receptors, which directly
impacts the drug design process \cite{Drug_design}.

It has been stressed by Berry \cite{Berry} that there are striking
similarities between organization schemes at different observation
scales in complex systems, such as allosteric-enzyme networks, cell
population and virus spreading. Recently, Miotto et al. \cite{Miotto}
exploited these similarities between epidemic spreading and a diffusive
process on a protein residue network to prove the capability of propagating
information in complex 3D protein structures. Their analogy
proved useful in estimating important protein properties ranging from
thermal stability to the identification of functional sites \cite{Miotto}.
In the current work, we go a step further in the exploitation of the
analogy between epidemiological models and communication processes
in proteins by considering the inclusion of long-range transmission
effects. For this purpose, we develop here a new fractional-order
Susceptible-Infected (SI) model for the transmission of perturbations
through the amino acids of a protein residue network. Such perturbations
are produced, for instance, by the interactions of the given protein
with inhibitors, such as drugs or drug candidates. We obtain an upper
bound to the exact soluction of this fractional-order SI model
which is expressed in terms of the Mittag-Leffler matrix functions,
and which generalizes the upper bound found by Lee et al.
\cite{Lee_et_al} to the non-fractional (classical) SI model.

Due to its current relevance, we apply the present approach to the
study of the long-range inter-residue communication in the main protease
of the new coronavirus named SARS-CoV-2 \cite{SARS-CoV-2,SARS-CoV-2_1}.
This new coronavirus has produced an outbreak of pulmonary disease
expanding from the city of Wuhan, Hubei province of China to the rest
of the World in about 3 months \cite{COVID-19}. One of the most important
targets for the development of drugs against SARS-CoV-2 is its main
protease, M\textsuperscript{pro}, whose 3-dimensional structure has
been recently resolved and deposited \cite{Crystal-Mpro} in the Protein
Data Bank (PDB) \cite{PDB}. It is a key enzyme for proteolytic processing
of polyproteins in the virus and some chemicals have been found to
bind this protein, representing potential specific drug canditades
against CoV-2 \cite{Crystal-Mpro}. Here we find that important communication
between amino acids in CoV-2 M\textsuperscript{pro} occurs from the
proximities of the binding site to very distant amino acids in other
domains of the protein. These effects produced by the interaction
with inhibitors are transmistted up to 50{\AA} away from the binding
site, confirming the long-range nature of intra-protein communication.
According to our results, it seems that stronger inhibitors transmit
such perturbations to longer inter-residue distances. Therefore, the
current findings are important for the understanding of the mechanisms
of drug action on CoV-2 M\textsuperscript{pro}, which may help to
the design of drug candidates against this new coronavirus.

\section{Antecedents and Motivations}

\subsection{Protein residue networks}

The protein residue networks (PRN) (see ref. \cite{Networks_1} Chapter
14 for details) are simple, undirected and connected graphs $G=\left(V,E\right)$, therefore their
adjacency matrices are symmetric matrices of order $n\times n$ and
have eigenvalues $\lambda_{1}>\lambda_{2}\geq\cdots\geq\lambda_{n}$.
As the matrices are traceless, the spectral radius $\lambda_{1}>0$.
Here $v_{i}\in V,$$\  i=1,\ldots,n$ are the nodes corresponding to the amino
acids of a protein and two nodes $v_{i}$ and $v_{j}$ are connected
by an edge $\left\{ v_{i},v_{j}\right\} \in E$ if the corresponding
residues (amino acids) interacts physically in the protein.  They are
built here by using the information reported on the Protein Data Bank
\cite{PDB} for the protease of CoV-2 as well as its complexes with
three inhibitors (see further). The nodes of the network represent
the \textgreek{a}-carbon of the amino acids. Then, we consider cutoff
radius $r_{C}$, which represents an upper limit for the separation
between two residues in contact. The distance $r_{ij}$ between two
residues $i$ and $j$ is measured by taking the distance between
C\textsubscript{\textgreek{a}} atoms of both residues. Then, when
the inter-residue distance is equal or less than $r_{C}$ both residues
are considered to be interacting and they are connected in the PRN.
The adjacency matrix $A$ of the PRN is then built with elements defined
by

\begin{equation}
A_{ij}=\left\{ \begin{array}{cc}
H\left(r_{C}-r_{ij}\right) & i\neq j,\\
0 & i=j,
\end{array}\right.
\end{equation}
where $H\left(x\right)$ is the Heaviside function which takes the
value of one if $x>0$ or zero otherwise. Here we use the typical
interaction distance between two amino acids, which is equal to 7.0
Å. We have tested distances below and over this threshold obtaining,
in general, networks which are either too sparse or too dense, respectively.
In this work we consider the structures of the free SARS CoV-2 main protease
with PDB code 6Y2E as well as the ones of the SARS CoV-2 with inhibitors
6M0K \cite{Inhibitors_11}, 6YZE \cite{Inhibitors_11} and 6Y2G \cite{Crystal-Mpro}.
For details of preprocessing the reader is directed to \cite{CHAOS}.

\subsection{Standard SI model}\label{Sect2.2}

Here we state the main motivation of using a Susceptible-Infected
(SI) model for studying the effects of inhibitor binding to a protein
residue network in a similar way as an SIS has been used by Miotto
et al. \cite{Miotto}. The selection of an SI model can be understood
by the fact that we are interested in the early times of the dynamics.
At this stage, it has been shown \cite{Lee_et_al} that the SI model
is most suitable than any other model. To motivate the SI model in
the PRN context let us consider that an amino acid is in the binding
site of a protein. Then, this amino acid is susceptible to be perturbed
by the interaction with this inhibitor. Consequently, this residue
can be in one of two states, either waiting to be perturbed (susceptible)
or being perturbed by the interaction. Of course, this amino acid
can transmit this perturbation to any other amino acid in the protein
to which it interacts with. Then, if $\beta$ is the rate at which
such perturbation is transmitted between amino acids, and if $s_{i}\left(t\right)$
and $x_{i}\left(t\right)$ are the probabilities that the residue
$i$ is susceptible or get perturbed at time $t$, respectively, we
can write the dynamics
\begin{equation}
\dfrac{ds_{i}\left(t\right)}{dt}=-\beta s_{i}\left(t\right)x_{i}\left(t\right),\label{eq31}
\end{equation}

\begin{equation}
\dfrac{dx_{i}\left(t\right)}{dt}=\beta s_{i}\left(t\right)x_{i}\left(t\right).\label{eq32}
\end{equation}

Because the amino acids can only be in the states ``susceptible''
or ``perturbed'' we have that $s_{i}\left(t\right)+x_{i}\left(t\right)=1$,
such that we can write

\begin{equation}
\dfrac{dx_{i}\left(t\right)}{dt}=\beta\left(1-x_{i}\left(t\right)\right)x_{i}\left(t\right).
\end{equation}

When we consider all the interactions between pairs of residues in
the PRN we should transform the previous equation into a system of
equations of the following form \cite{Bullo}:

\begin{equation}
\dfrac{dx_{i}\left(t\right)}{dt}=\beta\left(1-x_{i}\left(t\right)\right)\sum_{j\in\mathcal{N}}A_{ij}x_{j}\left(t\right),t\geq t_{0},\label{eq:SI_original}
\end{equation}
where $A_{ij}$ are the entries of the adjacency matrix of the PRN
for the pair of amino acids $i$ and $j$, and $\mathcal{N}=\{1,\ldots,n\}.$ In matrix-vector form becomes:

\begin{equation}
\dfrac{dx\left(t\right)}{dt}=\beta\left[I_{n}-\textnormal{diag}\left(x\left(t\right)\right)\right]Ax\left(t\right),
\end{equation}
with initial condition $x\left(0\right)=x_{0}$. The evolution of
dynamical systems based on the adjacency matrix of a network have been
analyzed by Mugnolo \cite{mugnolo}. It is well-known that \cite{Lee_et_al}:
\begin{enumerate}
\item if $x_{0}\in[0,1]^n$ then $x(t)\in[0,1]^n$ for all $t>0$;
\item $x(t)$ is monotonically non-decreasing in $t$;
\item there are two equilibrium points: $x^{\star}=0$, i.e. no epidemic,
and $x^{\star}=1$, i.e. full contagion;
\item the linearization of the model around the point 0 is given by

\begin{equation}
\dfrac{dx\left(t\right)}{dt}=\beta A\,x\left(t\right),\label{eq:model}
\end{equation}
and the solution diverges when $t\to \infty$, due to the fact that the spectral radius of $A$ is positive;
\item each trajectory with $x_{0}\neq0$ converges asymptotically to $x^{\star}=1$,
i.e. the epidemic spreads monotonically to the entire network.
\end{enumerate}
The SI model can be rewritten as

\begin{equation}
\dfrac{1}{1-x_{i}(t)}\dfrac{dx_{i}\left(t\right)}{dt}=\beta\sum_{j\in\mathcal{N}}A_{ij}\left(1-e^{-\left(-\log\left(1-x_{j}\left(t\right)\right)\right)}\right),
\end{equation}
which is equivalent to

\begin{equation}
\dfrac{dy_{i}\left(t\right)}{dt}=\beta\sum_{j\in\mathcal{N}}A_{ij}f\left(y_{j}\left(t\right)\right),
\end{equation}
where
\begin{equation}
y_{i}\left(t\right)\coloneqq g\left(x_{i}\left(t\right)\right)=-\log\left(1-x_{i}\left(t\right)\right)\in\left[0,\infty\right]\ ,\label{eq: g}
\end{equation}
 and $f\left(y\right)\coloneqq1-e^{-y}=g^{-1}\left(y\right)$.

Lee et al. \cite{Lee_et_al} have considered the following linearized
version of the previous nonlinear equation
\begin{equation}
\dfrac{d\hat{y}\left(t\right)}{dt}=\beta A\textnormal{diag}\left(1-x\left(t_{0}\right)\right)\hat{y}\left(t\right)+\beta b\left(x\left(t_{0}\right)\right),\label{eq: first linearization}
\end{equation}
where $\hat{x}\left(t\right)=f\left(\hat{y}\left(t\right)\right)$
in which $\hat{x}\left(t\right)$ is the approximate solution to the
SI model, $\hat{y}\left(t_{0}\right)=g\left(x\left(t_{0}\right)\right)$
and
\begin{equation}
    b\left(x\right)\coloneqq x+\left(1-x\right)\log\left(1-x\right).\label{eq: b}
\end{equation}
They have found that the solution to this linearized model is \cite{Lee_et_al}:

\begin{equation}
\begin{split}\hat{y}\left(t\right) & =e^{\beta\left(t-t_{0}\right)A\textnormal{diag}\left(1-x\left(t_{0}\right)\right)}g\left(x\left(t_{0}\right)\right)\\
 & +\sum_{k=0}^{\infty}\dfrac{\left(\beta\left(t-t_{0}\right)\right)^{k+1}}{\left(k+1\right)!}\left[A\textnormal{diag}\left(1-x\left(t_{0}\right)\right)\right]^{k}Ab\left(x\left(t_{0}\right)\right).
\end{split}
\end{equation}
When $t_{0}=0$, $x_{i}\left(0\right)=c/n$, $i=1,2,\ldots,n$ for
some positive $c$, the previous equation is transformed to
\begin{equation}
\hat{y}\left(t\right)=\left(1/\gamma-1\right)e^{\gamma\beta tA}\vec{1}-\left(1/\gamma-1+\log\left(\gamma\right)\right)\vec{1},
\end{equation}
where $\gamma=1-c/n$ and $\vec{1}$ is the all-ones vector. Note that the condition $x_{i}\left(0\right)=c/n$ indicates that at initial time every amino acid has the same probability of being perturbed by the inhibitor. Lee et
al. \cite{Lee_et_al} have proved that this solution is an upper bound
to the exact solution of the SI model. This result indicates that
the upper bound to the solution of the SI model is proportional to
the exponential of the adjacency matrix of the network, which is the
source of the subgraph centrality \cite{SC} and of the communicability
function \cite{communicability} between pairs of nodes in it. In
the next section of this work we obtain a generalization of this upper
bound based on a fractional-order SI model, which will also
be formulated there.

\section{Mathematical Results}

\subsection{Definition of the fractional-order SI model}

In the following we will consider a fractional SI model based on the
Caputo fractional derivative of the logarithmic function of $1-x_{i}.$ Here, $x_i$ also denotes the probability
that the residue $i$ get perturbed at time $t.$

First of all, we recall the definition of Caputo fractional derivative.
Given $0<\alpha<1$ and a function $u:[0,\infty)\to\mathbb{R}$, we
denote by $D_{t}^{\alpha}u$ the Caputo  fractional derivative of $u$
of order $\alpha,$ which is given by \cite{Caputo_1}
\[
D_{t}^{\alpha}u\left(t\right)=\int_{0}^{t}h_{1-\alpha}\left(t-\tau\right)u'\left(\tau\right)\,d\tau:=\left(h_{1-\alpha}*u'\right)\left(t\right),\quad t>0,
\]
where $*$ denotes the classical convolution product on $\left(0,\infty\right)$
and $h_{\gamma}\left(t\right)\coloneqq\frac{t^{\gamma-1}}{\Gamma(\gamma)},$
for $\gamma>0.$ Observe that the previous fractional derivative has
sense whenever the function is derivable and the convolution is defined
(for example if $u'$ is locally integrable). The notation $h_{\gamma}$
is very useful in the fractional calculus theory, mainly by the property
$h_{\gamma}*h_{\delta}=h_{\gamma+\delta}$ for all $\gamma,\delta>0.$

Before presenting our model, we state a technical lemma which plays
a key role in the main result of this section.
\begin{lemma} \label{lemma1}
 Let $u:[0,\infty)\to\mathbb{R}$ be a derivable function
with $u\left(0\right)=0,$ and $0<\alpha<1.$ If $D_{t}^{\alpha}u\left(t\right)\geq0$
for all $t>0,$ then $u\left(t\right)\geq0.$
\end{lemma}

\begin{proof}
Observe that by hypothesis $\left(h_{1-\alpha}*u'\right)\left(t\right)\geq0,$
therefore
\[
u\left(t\right)=\int_{0}^{t}u'\left(\tau\right)\,d\tau=\left(h_{1}*u'\right)\left(t\right)=\left(h_{\alpha}*h_{1-\alpha}*u'\right)\left(t\right)\geq0.
\]
\end{proof}
Now, we recall that $\beta$ will denote the perturbation rate and
let $s_{i}\left(t\right)$ and $x_{i}\left(t\right)$ be the probabilities
that residue $i$ is susceptible or get perturbed at time $t$, respectively.
Let $0<\alpha<1$, we consider the following fractional model inspired
by \eqref{eq31} and \eqref{eq32}:

\[
\left\{ \begin{array}{l}
{\displaystyle \int_{0}^{t}h_{1-\alpha}\left(t-\tau\right)\dfrac{s_{i}'\left(\tau\right)}{x_{i}\left(\tau\right)}\,d\tau=-\beta^{\alpha}s_{i}\left(t\right),}\\
\\
{\displaystyle \int_{0}^{t}h_{1-\alpha}\left(t-\tau\right)\dfrac{x_{i}'\left(\tau\right)}{s_{i}\left(\tau\right)}\,d\tau=\beta^{\alpha}x_{i}\left(t\right).}
\end{array}\right.
\]
Since  $s_{i}\left(t\right)+x_{i}\left(t\right)=1$, we have
\begin{equation}
\int_{0}^{t}h_{1-\alpha}\left(t-\tau\right)\dfrac{x_{i}'\left(\tau\right)}{1-x_{i}\left(\tau\right)}\,d\tau=\beta^{\alpha}x_{i}\left(t\right).
\end{equation}
Observe that the left-hand side of the above system is the Caputo fractional derivative of the minus logarithmic function of $1-x_{i}$ (see for instance \cite{Logarithmic}), that
is,
\[
D_{t}^{\alpha}(-\log(1-x_{i}))(t).
\]

As in the classical SI model happens, this equation is transformed
into a system of equations when we consider the interactions between
the different residues in the protein according to the PRN. So, the
fractional SI model which we will study is given by
\begin{equation}
\begin{array}{ll}
{\displaystyle \int_{0}^{t}h_{1-\alpha}\left(t-\tau\right)\dfrac{x_{i}'\left(\tau\right)}{1-x_{i}\left(\tau\right)}\,d\tau=\beta^{\alpha}\sum_{j\in\mathcal{N}}A_{ij}x_{j},} & i\in\mathcal{\mathcal{N}},\ t>0,\ x_{i}(0)\in[0,1].\end{array}\label{fractsystem}
\end{equation}
We can rewrite \eqref{fractsystem} in a matrix-vector form:
\begin{equation}\label{fractsystemmatrix}
D_{t}^{\alpha}(-\log(1-x))(t)=\beta^{\alpha}Ax\left(t\right),
\end{equation}
with initial condition $x\left(0\right)=x_{0},$ where $A$ is the adjacency matrix of the PRN. This fractional SI
model, based on the fractional-order derivative, has not been
considered in the literature under our knowledge. Other fractional
compartmental models have been previously discussed in the literature
(see for instance \cite{Fractional_compartments} and references therein).

Note that if $x_{i}(0)=1,$ then by \eqref{fractsystem} we have $\dfrac{x_{i}'\left(\tau\right)}{1-x_{i}\left(\tau\right)}\geq0$
for $s$ close to 0, and that case is not possible. So, we will consider
that $x_{i}^{\star}=1$ is an equilibrium point. The same happens
if $x_{i}^{\star}=0$ by the equations given for $s_{i}.$ Furthermore,
if $x_{i}(0)\in(0,1),$ by Lemma~\ref{lemma1} we have $-\log(1-x_{i}(t))\geq-\log(1-x_{i}(0))>0,$
then $x_{i}(t)\in(0,1),$ and therefore $x_{i}$ is non-decreasing.
We deduce that if $x(0)\in[0,1]^{n}$ then $x(t)\in[0,1]^{n}$ for
all $t>0,$ and there are two equilibrium points: $x^{\star}=0$,
i.e. no epidemic, and $x^{\star}=1$, i.e., full contagion. Also,
each trajectory with $x_{0}\neq0$ converges asymptotically to $x^{\star}=1$,
i.e. the epidemic spreads monotonically to the entire network.

One of the objects of greatest importance in the fractional calculus theory are the Mittag-Leffler functions. Let $\alpha,\nu>0,$ they are defined by
\begin{equation}
E_{\alpha,\nu}(z)=\sum_{k=0}^\infty\frac{z^k}{\Gamma(\alpha k+\nu)}\ ,\ z\in\mathbb{C} .
    \label{eq: ML definition}
\end{equation}
For more details on fractional calculus and Mittag-Leffler functions see the seminal works  \cite{diethelm2010analysis,gorenflo2001fractional,mainardi2000mittag,paris2002exponential,gorenflo2014mittag}. Let us note that when $\alpha=1$ this function reduces to $e^z$.  As the exponential function, the Mittag-Leffler functions can be considered in a matrix framework. We refer the reader to Section \ref{sec: ML matrix functions} for more details on the Mittag-Leffler matrix functions.

Now we consider the linearization of \eqref{fractsystemmatrix}
\begin{equation}\label{fractsystemlinear}
D_{t}^{\alpha}\tilde{x}(t)=\beta^{\alpha}A\tilde{x}(t).
\end{equation}
It is known that the solution of \eqref{fractsystemlinear} is given by
\begin{equation}
\tilde{x}(t)=E_{\alpha,1}\left((\beta t)^{\alpha}A\right)x_{0}:=\sum_{k=0}^{\infty}\frac{(\beta t)^{\alpha k}A^{k}x_{0}}{\Gamma\left(\alpha k+1\right)},\label{tilde}
\end{equation}
where $x_{0}$ is the same initial condition that in the non-linearized
problem. In fact the solution diverges as $t$ goes to infinity, that is,
\begin{equation}
\begin{split}\lim_{t\rightarrow\infty}\tilde{x}_{i}\left(t\right) & =\lim_{t\rightarrow\infty}E_{\alpha,1}\left(\left(\beta t\right)^{\alpha}\lambda_{1}\right)\psi_{1i}\sum_{j=1}^{n}\psi_{1j}x_{0j}\\
 & =\lim_{t\rightarrow\infty}\sum_{k=0}^{\infty}\dfrac{\left(\left(\beta t\right)^{\alpha}\lambda_{1}\right)^{k}}{\varGamma\left(\alpha k+1\right)}\psi_{1i}\sum_{j=1}^{n}\psi_{1j}x_{0j}\\
 & =\infty,
\end{split}
\end{equation}
for all $v_i\in V$ in $G=\left(V,E\right),$ and where $\psi_{1j}$
is the $j$th entry of the eigenvector associated to the spectral
radius $\lambda_{1}$.

Observe that the fractional SI model \eqref{fractsystem} can be rewritten as
\[
D_{t}^{\alpha}y_{i}\left(t\right)=\beta^{\alpha}\sum_{j\in\mathcal{N}}A_{ij}f\left(y_{j}\left(t\right)\right),
\]
where $y_i(t)$ is defined as in (\ref{eq: g}).

Now we consider the Lee-Tenneti-Eun (LTE) type transformation \cite{Lee_et_al}, which is also given in \eqref{eq: first linearization},
which produces the following linearized equation
\begin{equation}
D_{t}^{\alpha}\hat{y}\left(t\right)=\beta^\alpha\hat{A}\hat{y}(t)+\beta^{\alpha}Ab\left(x_0\right),\label{eqlineal}
\end{equation}
where $\hat{A}=A\Omega$ and
 $\Omega:=\text{diag}\left(1-x_{0}\right)$. Analogous to the notation used in (\ref{eq: first linearization}), $\hat{x}\left(t\right)=f\left(\hat{y}\left(t\right)\right)$
in which $\hat{x}\left(t\right)$ is an approximate solution to the
fractional SI model, $\hat{y}$ is the solution of \eqref{eqlineal}
with initial condition $\hat{y}\left(0\right)=g\left(x\left(0\right)\right)$
and $b\left(x_0\right)$ is given in (\ref{eq: b}).
\begin{theorem}\label{theorem1}
For any $t\geq0$, we have
\[
x(t)\preceq\hat{x}(t)=f(\hat{y}(t))\preceq\tilde{x}(t),
\]
under the same initial conditions $x_{0}:=x(0)=\hat{x}(0)=\tilde{x}(0),$ where ${x}(t)\preceq\hat{{x}}(t)$ if $x_i\leq \hat{x}_i$ for all $i=1,\ 2,\ \dotsc,\ n$.
The solution $\hat{y}$ of \eqref{eqlineal} is given by
\begin{equation}
\hat{y}\left(t\right)=E_{\alpha,1}\left((\beta t)^{\alpha}\hat{A}\right)g\left(x_{0}\right)+\sum_{k=0}^{\infty}\frac{(\beta t)^{\alpha\left(k+1\right)}\hat{A}^{k}Ab\left(x_{0}\right)}{\Gamma\left(\alpha\left(k+1\right)+1\right)},\label{y}
\end{equation}
and $\tilde{x}$ is given by \eqref{tilde}. Furthermore, $\|\hat{x}(t)-x(t)\|\to0$
and $\|\tilde{x}(t)-x(t)\|\to\infty$ as $t$ goes to infinity.
\end{theorem}

\begin{proof}
First of all, by the theory of fractional calculus, it is well-known
that the solution of the linearized problem \eqref{eqlineal} is given by
\begin{equation}
\hat{y}\left(t\right)=E_{\alpha,1}\left((\beta t)^{\alpha}\hat{A}\right)g\left(x_{0}\right)+\int_{0}^{t}\tau^{\alpha-1}E_{\alpha,\alpha}\left((\beta \tau)^{\alpha}\hat{A}\right)\beta^{\alpha}Ab\left(x_0\right)\,d\tau,\label{soly}
\end{equation}
where the functions $E_{\alpha,1}(\cdot)$ and $E_{\alpha,\alpha}(\cdot)$ are defined as in (\ref{eq: ML definition}).
Therefore, since
\[
\int_{0}^{t}\tau^{\alpha k+\alpha-1}\,d\tau=\frac{t^{\alpha k+\alpha}}{\alpha k+\alpha},
\]
from \eqref{soly} we get \eqref{y}. For more details about linear
fractional models see \cite{AA,ALM,B}, and references therein. Notice
that Eq. (\ref{y}) is the generalized fractional version of the one
obtained by LTE by means of their Theorem \ref{theorem1}. Their specific solution
is recovered when $\alpha=1$ where $E_{1,1}\left(\beta t\hat{A}\right)=\exp\left(\beta t\hat{A}\right)$
and $\Gamma\left(n+2\right)=\left(n+1\right)!$.

We have assumed that $x_{0}=x\left(0\right)=\hat{x}\left(0\right)=\tilde{x}\left(0\right),$
with $y\left(t\right)=g\left(x\left(t\right)\right)$ and $\hat{y}\left(t\right)=g\left(\hat{x}\left(t\right)\right).$
Since $y,\hat{y}$ are non-decreasing functions of $x,\hat{x},$ it
is enough to prove that $y\left(t\right)\preceq\hat{y}\left(t\right)$
to get $x\left(t\right)\preceq\hat{x}\left(t\right).$ Following the
paper of Lee et all, since $f$ is a concave function with $f'\left(y\right)=e^{-y},$
we have
\[
D_{t}^{\alpha}y_{i}\left(t\right)\leq\beta^{\alpha}\sum_{j\in\mathcal{N}}A_{ij}\left(1-x_{j}\left(0\right)\right)y_{j}\left(t\right)+\beta^{\alpha}\sum_{j\in\mathcal{N}}A_{ij}b\left(x_{j}\left(0\right)\right).
\]
Then, since $y\left(0\right)=\hat{y}\left(0\right),$ $D_{t}^{\alpha}y\left(t\right)\preceq D_{t}^{\alpha}\hat{y}\left(t\right),$
so Lemma \ref{lemma1} implies $x\left(t\right)\preceq\hat{x}\left(t\right).$

Now, note that
\[
D_{t}^{\alpha}\hat{x}_{i}\left(t\right)=D_{t}^{\alpha}f\left(\hat{y}_{i}\left(t\right)\right)=\int_{0}^{t}h_{1-\alpha}\left(t-s\right)e^{-\hat{y}_{i}\left(s\right)}\hat{y}_{i}'\left(s\right)\,ds.
\]
Furthermore (\ref{y}) shows that $y_{i}'\left(s\right)\geq0$ for
all $s>0,$ then
\[
0\leq D_{t}^{\alpha}\hat{x}_{i}\left(t\right)\leq\int_{0}^{t}h_{1-\alpha}\left(t-s\right)\hat{y}_{i}'\left(s\right)\,ds=D_{t}^{\alpha}\hat{y}_{i}\left(t\right).
\]

Also, it is well-known  \cite{diethelm2010analysis,gorenflo2014mittag,podlubny1998fractional} (or more recently \cite{AA,ALM,B}) that the
previous Mittag-Leffler matrix functions satisfy
\begin{align}
E_{\alpha,1}\left((\beta t)^{\alpha}\hat{A}\right)&=\left(h_{1-\alpha}*s^{\alpha-1}E_{\alpha,\alpha}\left((\beta s)^{\alpha}\hat{A}\right)\right)\left(t\right)\nonumber\\ &=\int_{0}^{t}h_{1-\alpha}\left(t-s\right)s^{\alpha-1}E_{\alpha,\alpha}\left((\beta s)^{\alpha}\hat{A}\right)\,ds\label{eq4}
\end{align}
and
\begin{align}
E_{\alpha,1}\left((\beta t)^{\alpha}\hat{A}\right)I&=I+\beta^{\alpha}\hat{A}\left(h_{\alpha}*E_{\alpha,1}\left((\beta s)^{\alpha}\hat{A}\right)\right)\left(t\right)\nonumber\\&=I+\beta^{\alpha}\hat{A}\int_{0}^{t}h_{\alpha}\left(t-s\right)E_{\alpha,1}\left((\beta s)^{\alpha}\hat{A}\right)\,ds.\label{eq5}
\end{align}
Then, by (\ref{eqlineal}), (\ref{soly}), (\ref{eq4}) and (\ref{eq5})
one gets
\begin{eqnarray*}
D_{t}^{\alpha}\hat{y}\left(t\right) & = & \beta^{\alpha}\hat{A}E_{\alpha,1}\left((\beta t)^{\alpha}\hat{A}\right)g\left(x_{0}\right)+\beta^{\alpha}\hat{A}\left(h_{1}*s^{\alpha-1}E_{\alpha,\alpha}\left((\beta s)^{\alpha}\hat{A}\right)\right)\left(t\right)\\
 & \times & \beta^{\alpha}Ab\left(x_0\right)+\beta^{\alpha}Ab\left(x_0\right)\\
 & = & \beta^{\alpha}\hat{A}E_{\alpha,1}\left((\beta t)^{\alpha}\hat{A}\right)g\left(x_{0}\right)+\beta^{\alpha}\hat{A}\left(h_{\alpha}*h_{1-\alpha}*s^{\alpha-1}E_{\alpha,\alpha}\left((\beta s)^{\alpha}\hat{A}\right)\right)\left(t\right)\\
 & \times & \beta^{\alpha}Ab\left(x_0\right)+\beta^{\alpha}Ab\left(x_0\right)\\
 & = & \beta^{\alpha}\hat{A}E_{\alpha,1}\left((\beta t)^{\alpha}\hat{A}\right)g\left(x_{0}\right)+\beta^{\alpha}\hat{A}\left(h_{\alpha}*E_{\alpha,1}\left((\beta s)^{\alpha}\hat{A}\right)\right)\\
 & \times & \left(t\right)\beta^{\alpha}Ab\left(x_0\right)+\beta^{\alpha}Ab\left(x_0\right)\\
 & = & \beta^{\alpha}\hat{A}E_{\alpha,1}\left((\beta t)^{\alpha}\hat{A}\right)g\left(x_{0}\right)+E_{\alpha,1}\left((\beta t)^{\alpha}\hat{A}\right)\beta^{\alpha}Ab\left(x_0\right)\\
 & = & \beta^{\alpha}AE_{\alpha,1}\left((\beta t)^{\alpha}\hat{A}\right)x_0,
\end{eqnarray*}
where in the last equality we have used that $\Omega g\left(x_0\right)+b\left(x_0\right)=x_0.$
By definition of Mittag-Leffler matrix function, it is easy to see
that
\[
E_{\alpha,1}\left((\beta t)^{\alpha}\hat{A}\right)x_0\preceq E_{\alpha,1}\left((\beta t)^{\alpha}A\right)x_0,
\]
since $\hat{A}=A\Omega$ with $\Omega=\text{diag}\left(1-x_0\right).$
Therefore
\[
D_{t}^{\alpha}\hat{x}\left(t\right)\preceq D_{t}^{\alpha}\hat{y}\left(t\right)\preceq\beta^{\alpha}AE_{\alpha,1}\left((\beta t)^{\alpha}A\right)x_0=D_{t}^{\alpha}\tilde{x}\left(t\right),
\]
and Lemma \ref{lemma1} implies $\hat{x}\left(t\right)\preceq\tilde{x}\left(t\right).$

Finally, it is known that $\lim_{t\to\infty}\tilde{x}_{i}\left(t\right)=\infty$
and $\lim_{t\to\infty}\hat{y}_{i}\left(t\right)=\infty.$ Since $f$
is continuous, $\lim_{t\to\infty}\hat{y}_{i}\left(t\right)=\infty,$
then $\lim_{t\to\infty}\hat{x}_{i}\left(t\right)=\lim_{t\to\infty}f\left(\hat{y}_{i}\right)\left(t\right)=1.$
Therefore, since $\lim_{t\to\infty}x_{i}\left(t\right)=1$ we conclude
$\|\hat{x}\left(t\right)-x\left(t\right)\|\to0$ and $\|\tilde{x}\left(t\right)-x\left(t\right)\|\to0$
as $t\to\infty.$
\end{proof}
\begin{corollary}
Let $x_{0}\preceq1$, then the solution of \eqref{eqlineal} can be
written as
\begin{equation}
\hat{y}\left(t\right)=g\left(x_{0}\right)+\left[E_{\alpha,1}\left((\beta t)^{\alpha}\hat{A}\right)-I\right]\Omega^{-1}x_0.
\end{equation}
\end{corollary}

\begin{proof}
Let us write Eq. (\ref{soly}) in the following way

\begin{equation}
\hat{y}\left(t\right)=E_{\alpha,1}\left((\beta t)^{\alpha}\hat{A}\right)g\left(x_{0}\right)+\int_{0}^{t}s^{\alpha-1}E_{\alpha,\alpha}\left((\beta s)^{\alpha}\hat{A}\right)\beta^{\alpha}A\Omega\Omega^{-1}b\left(x_0\right)\,ds,
\end{equation}
which can be reordered as

\begin{equation}
\hat{y}\left(t\right)=E_{\alpha,1}\left((\beta t)^{\alpha}\hat{A}\right)g\left(x_{0}\right)+\left[\beta^{\alpha}\hat{A}\int_{0}^{t}s^{\alpha-1}E_{\alpha,\alpha}\left((\beta s)^{\alpha}\hat{A}\right)\,ds\right]\Omega^{-1}b\left(x_0\right).
\end{equation}
So, by \eqref{eq5} we have
\begin{equation}
\hat{y}\left(t\right)=E_{\alpha,1}\left((\beta t)^{\alpha}\hat{A}\right)g\left(x_{0}\right)+\left[E_{\alpha,1}((\beta t)^{\alpha}\hat{A})-I\right]\Omega^{-1}b\left(x_0\right).
\end{equation}

Now, it is easy to check that $\Omega^{-1}b\left(x_0\right)=\Omega^{-1}x_0-g\left(x_{0}\right).$
Therefore,

\begin{equation}
\hat{y}\left(t\right)=E_{\alpha,1}\left((\beta t)^{\alpha}\hat{A}\right)\Omega^{-1}x_0-\Omega^{-1}x_0+g\left(x_{0}\right),
\end{equation}
which by reordering gives the final solution.
\end{proof}
Let us now consider $x_{i}\left(0\right)=c/n$, $i=1,2,\ldots,n,$ where $c$ is a positive real, and let $\gamma=1-c/n.$ Noting that $\textnormal{diag}\left(1-x_0\right)=\gamma I$,
then

\begin{align}
\hat{y}\left(t\right) & =\left(\frac{1-\gamma}{\gamma}\right)E_{\alpha,1}\left(t^{\alpha}\beta^{\alpha}\hat{A}\right)\vec{1}-\left(\frac{1-\gamma}{\gamma}+\log\gamma\right)\vec{1}\nonumber \\
 & =\left(\frac{1-\gamma}{\gamma}\right)E_{\alpha,1}\Bigl(t^{\alpha}\beta^{\alpha}A\textnormal{diag}\left(1-x_0\right)\Bigr)\vec{1}-\left(\frac{1-\gamma}{\gamma}+\log\gamma\right)\vec{1}\nonumber \\
 & =\left(\frac{1-\gamma}{\gamma}\right)E_{\alpha,1}\Bigl(t^{\alpha}\beta^{\alpha}\gamma A\Bigr)\vec{1}-\left(\frac{1-\gamma}{\gamma}+\log\gamma\right)\vec{1}.
\end{align}

We should remark that according to the result in Theorem \ref{theorem1} the
solution to the fractional-order SI model obtained here
represents an upper bound to the exact solution. Therefore, we will use
it here as the worse-case scenario for the analysis of perturbations in
real-world PRNs. This means that our results should be interpreted here
not as an approximation to the solution but as the most extreme
situation that can happen in the propagation of a perturbation through a
protein.

\subsection{Why are fractional derivatives needed to study PRNs?}\label{sec: ML matrix functions}

As we have seen in Section \ref{Sect2.2} the upper bound
of the SI model is linearly proportional to $e^{\alpha\beta tA}\vec{1}$,
where $A$ is the adjacency matrix of the graph. That is, the only
structural information about the graph which appears in the solution of
the SI model is contained in $e^{\zeta A}$ where $\zeta$ is a parameter.
Here we first explain how is this information encoded in the matrix
exponential. Let us start by writing
\begin{equation}
e^{\zeta A}=\sum_{k=0}^{\infty}\dfrac{\left(\zeta A\right)^{k}}{k!}.
\end{equation}

\textcolor{black}{We recall that a }\textit{\textcolor{black}{walk}}\textcolor{black}{{}
of length $k$ in $G$ is a set of nodes $i_{1},i_{2},\ldots,i_{k},i_{k+1}$
such that for all $1\leq l\leq k$, $(i_{l},i_{l+1})\in E$. A }\textit{\textcolor{black}{closed
walk}}\textcolor{black}{{} is a walk for which $i_{1}=i_{k+1}$ \cite{Networks_1}.} Then, we state the following well-known result (see \cite{Networks_1} and references therein).
\begin{theorem}
The number of walks of length $k$ between the nodes $u$ and $v$
of the graph $G$ is given by $\left(A^{k}\right)_{uv}.$
\end{theorem}

This means that $\left(e^{\zeta A}\right)_{uv}$ counts the number of walks
of any length between the nodes $u$ and $v$ of $G$ penalizing them
by $\left(k!\right)^{-1},$ where $k$ is the length of the walk. Obviously,
$\left(e^{\zeta A}\right)_{uv}$ penalizes too heavily relatively
long walks. For instance, while a walk of length two contributes $0.5\zeta^{2}$
to $\left(e^{\zeta A}\right)_{uv}$, a walk of length 6 contributes $0.0014\zeta^{6}$. Then, if $\zeta<1$ the last contribution is practically
null.

It is known that the transmission of perturbations through a PRN is characterized by
two main properties:
\begin{enumerate}
\item there is a wide range of time frames, ranging from $10^{-3}$ seconds for conformational
transitions to $10^{-12}$ seconds for hydrogen bond breaking, rotational
relaxation and translational diffusion \cite{NMR};
\item the existence of long-range transmission of effects, which has been observed to
take place even between amino acids separated 100 Å apart \cite{Ottemann}.
Notice that in terms of the PRN this represents a transmission between
two nodes separated by 14 edges in the network.
\end{enumerate}
The function $e^{\zeta A}$ along cannot account for the previously
mentioned important characteristics of protein perturbations. Once
we consider a given network and a fixed value of $\zeta,$ the function
$e^{\zeta A}$ can describe only one process in the wide time-window
previously described. For instance, suppose that such process is one occurring at the $10^{-3}$ seconds scale. For the same network and conditions we cannot model another process occurring at the $10^{-10}$ seconds scale with the same mathematical model. At the same time this function penalizes
very heavily the long-range transmission of perturbation effects,
also avoiding a complete characterization of the physico-chemical
process.

In contrast, the Mittag-Leffler matrix functions, which appear in
the solution of the fractional-order SI model, are expressed in the
following way \cite{ML-1,Garrappa_1,ML-3,ML-4}
\begin{equation}
E_{\alpha,\nu}\Bigl(\zeta A\Bigr)=\sum_{k=0}^{\infty}\dfrac{\left(\zeta A\right)^{k}}{\Gamma\left(\alpha k+\nu\right)},\quad \alpha,\nu>0,.
\end{equation}
Then, for a fixed network topology and fixed external conditions $\zeta$
we can still model several processes at different time-windows by changing the Mittag-Leffler
parameter $\alpha$. For instance, we can consider a process occurring at the micro-second scale modeled by using $\alpha=1.0$, while another process occurring in the same network at the pico-second scale by using $\alpha=0.25$. This is illustrated in Fig. \ref{relative}(a) were
we plot the time evolution of the propagation of perturbations on
a cycle of 15 nodes for $\zeta=1$. As can be seen the time at which
50\% of the nodes are perturbed changes from 242 with $\alpha=1$
to 21 for $\alpha=0.25$. This simple graph, a cycle, is a good example
of some structures appearing in PRNs, named the chordless cycles or
holes. A chordless cycle, also known as induced cycle, is a cycle
which contains no edge which does not itself belongs to the cycle.
Holes are ubiquitous in proteins \cite{PRN_2} and they may represent
important binding sites in them.

The Mittag-Leffler matrix functions also allow to describe the second characteristic
of the propagation of perturbations through proteins, i.e., the existence
of long-range interactions. While $\left(e^{\zeta A}\right)_{uv}$
penalizes very heavily long-range perturbations, $E_{\alpha,1}\Bigl(\zeta A\Bigr)$
allows us to modulate such effects by changing the parameter $\alpha$.
For instance, let us consider a perturbation at a given node of the
cycle of 15 nodes previously considered here. This perturbation can
be transmitted across the cycle in no more than 7 steps, i.e., the
diameter of the graph. When $\alpha=1$, i.e., $E_{1,1}\Bigl(\zeta A\Bigr)=\exp\Bigl(\zeta A\Bigr)$,
the transmission of this perturbation to nodes at more than 5 steps
from the origin is practically null. As can be seen in Fig. \ref{relative}(b)
this situation changes when we drop the value of $\alpha$. When $\alpha=0.5$
we have 10\% of transmission to the farthest node relative to the transmission
to the nearest neighbors. When $\alpha=0.25$ the transmission to farthest
neighbors is almost unchanged in relation to that of the transmission
to nearest neighbors, which may seem exaggerated in physical conditions
of proteins.

\begin{figure}
\begin{centering}
\subfloat[]{\begin{centering}
\includegraphics[width=0.45\textwidth]{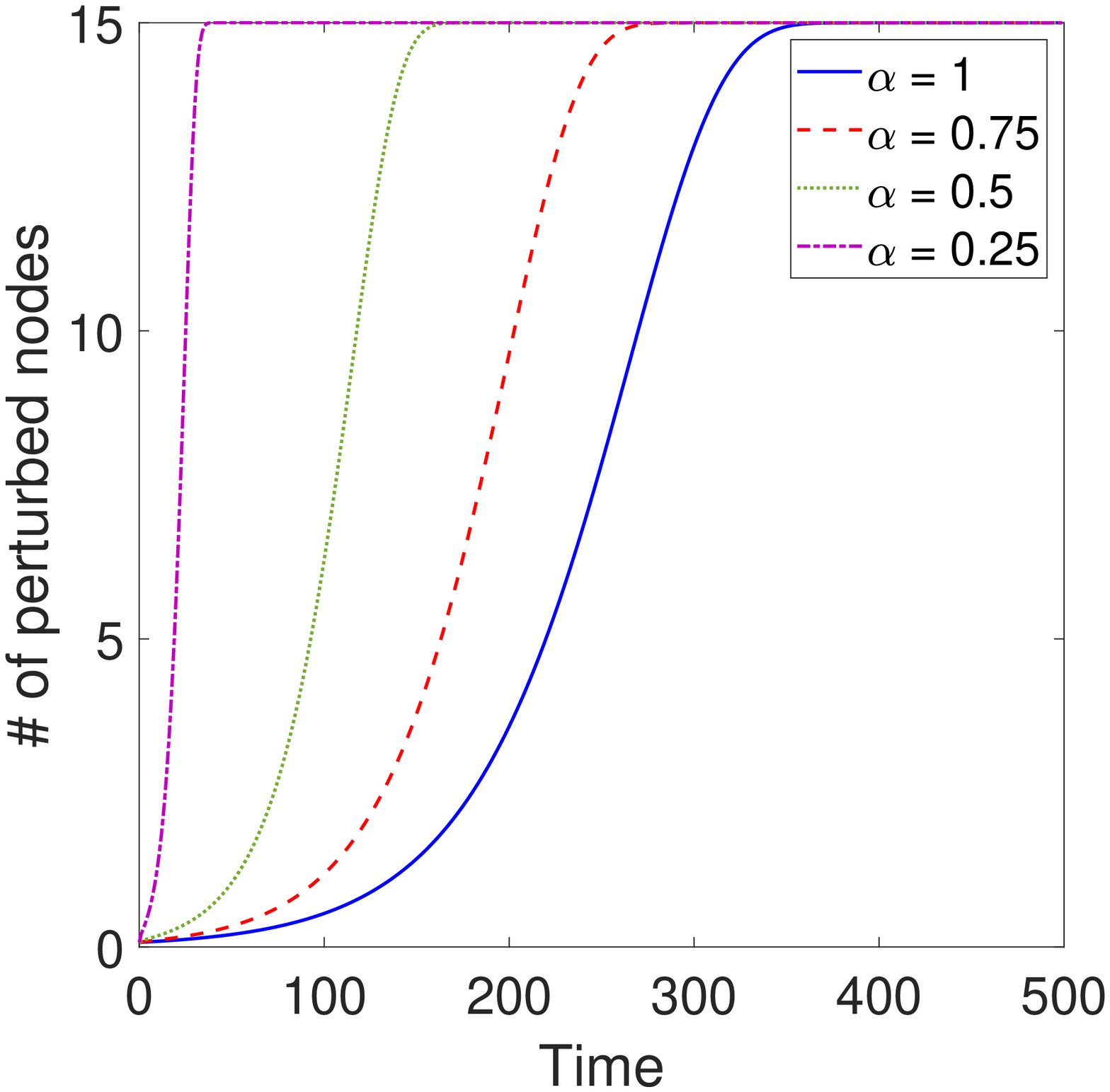}
\par\end{centering}

}\subfloat[]{\begin{centering}
\includegraphics[width=0.45\textwidth]{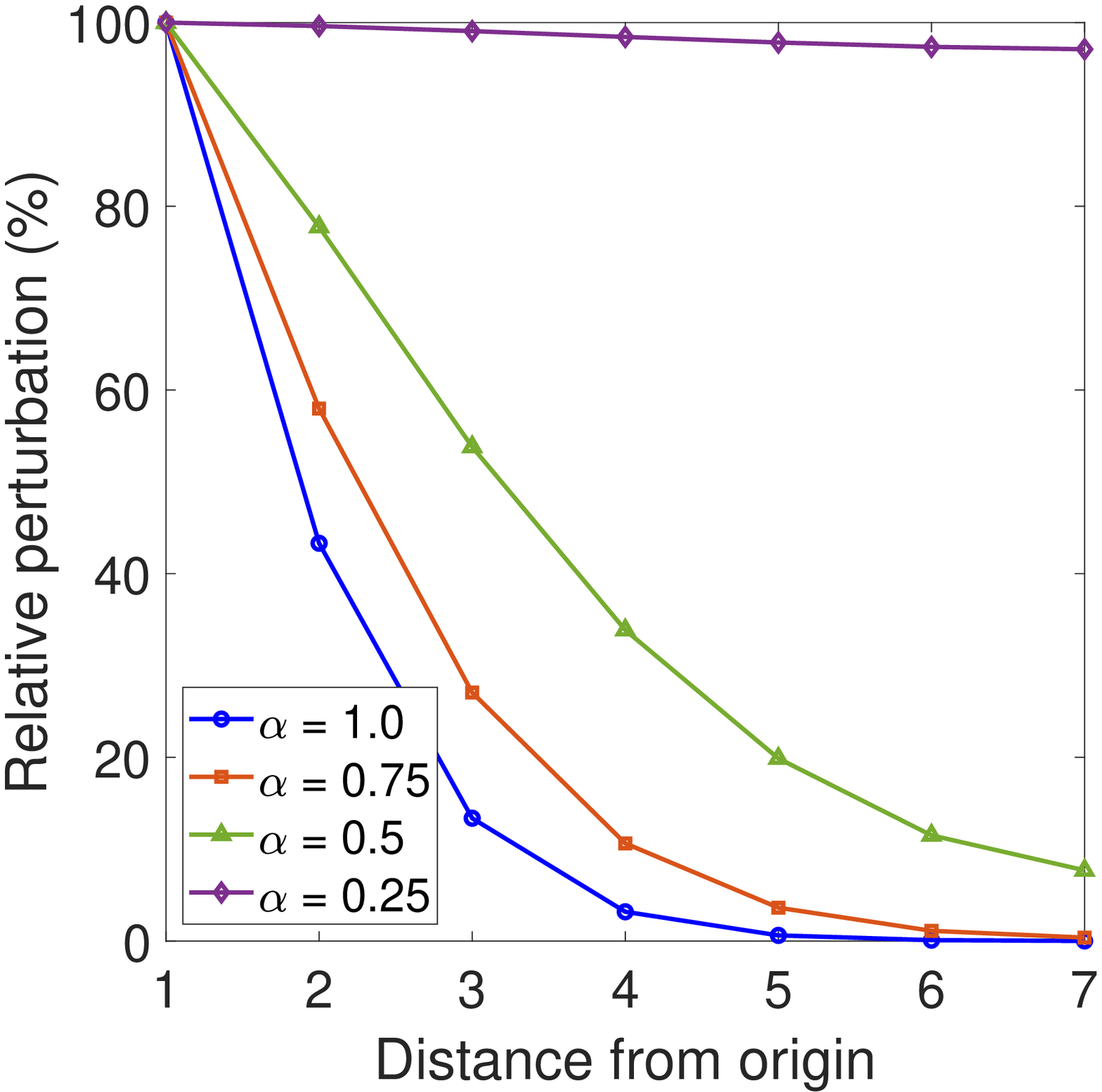}
\par\end{centering}
}
\par\end{centering}
\caption{Illustration of the effects of changing the parameter $\alpha$ in
the Mittag-Leffler function for the transmission of perturbations
in a cycle of 15 nodes.}

\label{relative}
\end{figure}

In closing, the Mittag-Leffler matrix functions, and consequently
the use of a fractional-order SI model, are important for modeling
the transmission of perturbations across PRNs because they allow to
capture important spatial and temporal characteristics of protein
perturbations, which are limited with the use of the classical SI
model.

\section{Computational Results}

Here we apply our model to the study of the M\textsuperscript{pro}
of SARS CoV-2 complexed with three inhibitors: PDB codes 6M0K and
6LZE from \cite{Inhibitors_11} and 6Y2G from \cite{Crystal-Mpro}.
We compare the results obtained with the free protease structure:
PDB 6Y2E. All calculations
are carried out on Matlab. For the Mittag-Leffler matrix functions
we use the Matlab function ``ml\_matrix.m'' provided by Garrappa
and Popolizio \cite{Garrappa_1,Garrappa_2}. The three inhibitors selected for our study have been reported
to display potent inhibitory capacity against SARS CoV-2. This potency
is measured through their inhibitory concentration $IC_{50}$, which
is the concentration of the inhibitor needed in vitro to inhibit the
virus by 50\%. For the simulations we use here $\beta=0.01$, $c=0.005,$ $\gamma=1-\frac{c}{n},$
and compare the results for $\alpha=\frac{1}{2}$ and when $\alpha=1$.
In Fig. \ref{Evolution} we illustrate the time evolution of the number
of perturbed amino acids in the complexes studied as well as in the
free protease (the last curve is overlapped by that of complex with
6LZE). There are two interesting observations from these
plots. First, the use of $\alpha=1/2$ produces a 10-fold
reduction of the time needed to reach the steady state of the process,
i.e., to perturb 100\% of the amino acids in the protease. The second
is that the order at which the different complexes reaches 50\% of
the amino acids perturbed is: 6M0K$ <$6LZE$<$6Y2G for both values of $\alpha,$ which is exactly the order of potency of the inhibitors towards SARS CoV-2.

\begin{figure}
\begin{centering}
\subfloat[]{\begin{centering}
\includegraphics[width=0.45\textwidth]{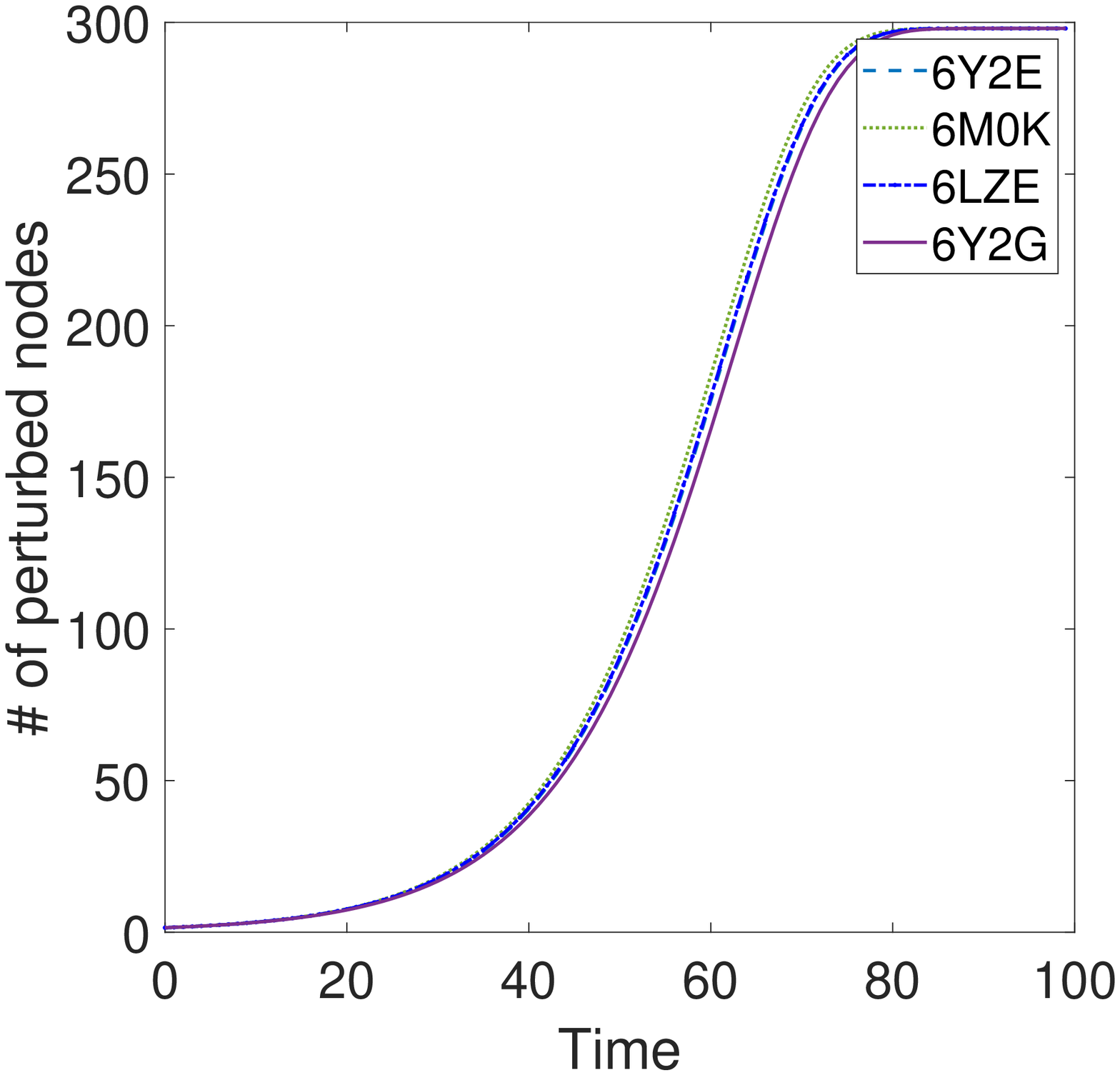}
\par\end{centering}
}\subfloat[]{\begin{centering}
\includegraphics[width=0.45\textwidth]{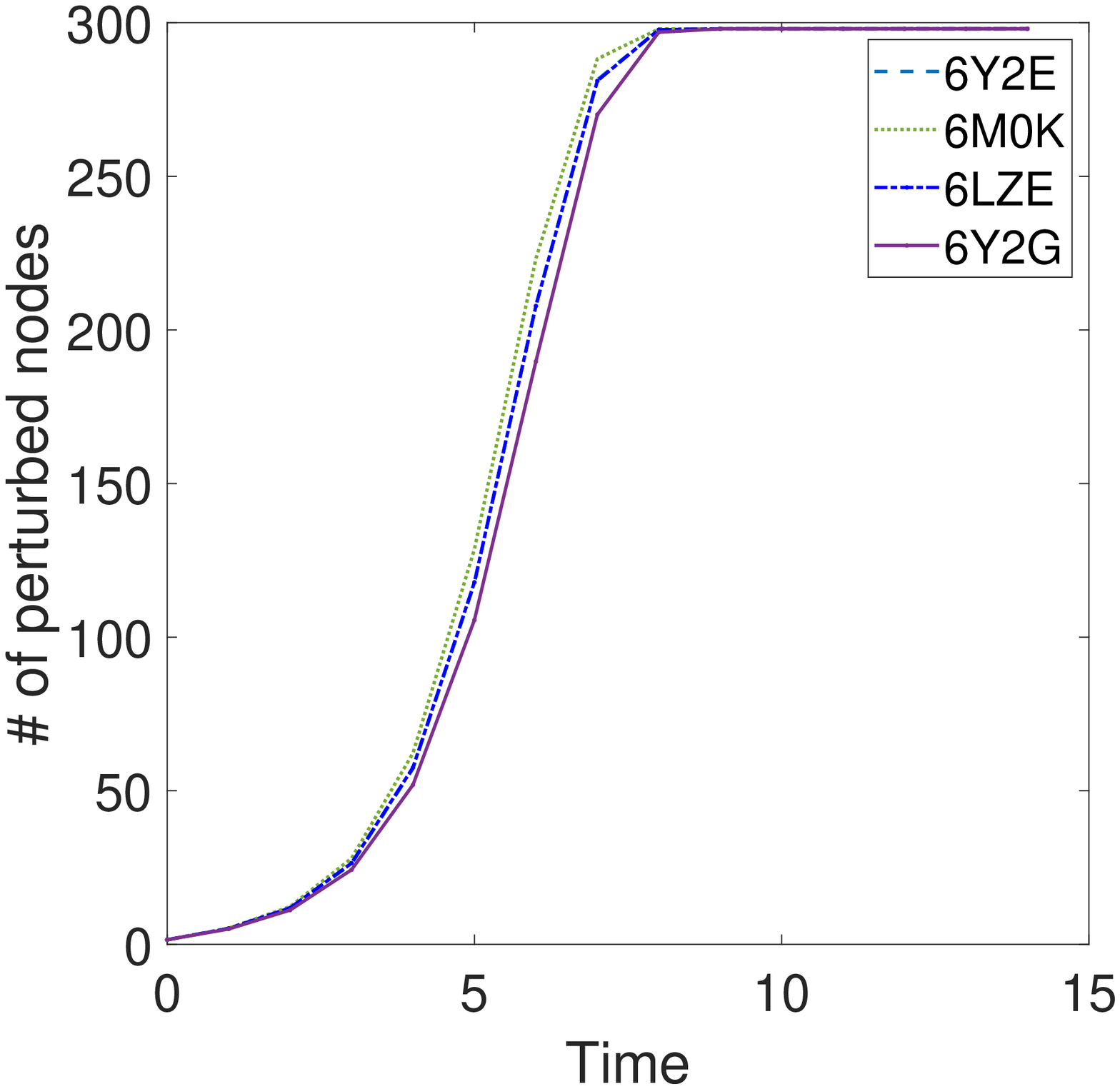}
\par\end{centering}
}
\par\end{centering}
\caption{Time evolution of the upper bounds of the normal (a) and fractional
(b) SI model for the main protease of CoV-2 bounded to three inhibitors
as well as free with $\beta=0.01$, $\gamma=1-\frac{c}{n},$ $c=0.005.$}

\label{Evolution}
\end{figure}

In order to gain more insights about the influence of the two different
dynamics on the propagation of a perturbation across the CoV-2 M\textsuperscript{pro}
when bounded with inhibitors we study the structural contributions
from each of the structures to the SI dynamics. In doing so we calculate
the relative differences in the individual components of the transmissibility
of this perturbation from one residue to another, $G_{ij}^{\alpha}$,

\begin{equation}
\varDelta G_{ij}^{\alpha}=\dfrac{1}{n\left(n-1\right)}\sum_{i\neq j}\dfrac{G_{ij}^{\alpha}\left(\textnormal{bounded}\right)-G_{ij}^{\alpha}\left(\textnormal{free}\right)}{G_{ij}^{\alpha}\left(\textnormal{free}\right)},
\end{equation}
where $G_{ij}^{\alpha}=[E_{\alpha,1}\biggl( (\beta t)^{\alpha}\gamma A \biggr)]_{ij}.$

We have selected the time at which 50\% of the amino acids in the
protease are perturbed, which occurs at $t=6$ ($\alpha=1/2$) and
$t=50$ ($\alpha=1$) for doing the calculations. The rest of the
parameters remain the same for both descriptors, i.e., $\beta=0.01$,
$\gamma=1-\frac{c}{n},$ $c=0.005.$ We also selected the top ten
pairs of amino acids according to their values of $\varDelta G_{ij}$.
For these pairs we have calculated the average length $\bar{L}$ of
the shortest paths connecting the pair of residues. For instance,
in 6M0K for $\alpha=1.0$ the largest value of $\varDelta G_{ij}$
is for the pair L167-K269 for which the shortest path has length 8.
For the same complex but using $\alpha=0.5$ the largest value of
$\varDelta G_{ij}$ is for the pair L167-M276 for which the shortest
path has length 10. In addition, we determine which $N_{BS}$ of these
pairs of residues in the top ten ranking according to $\varDelta G_{ij}$ is involved directly in the binding site
of the protease or it is bounded to one of them. For instance, for
the case of the pair before mentioned for 6M0K ($\alpha=1.0$) the
residue L167 is directly bounded to two amino acids in the binding
site, namely E166 and P168.

\begin{table}
\begin{centering}
\begin{tabular}{|c||c|c|c|c|c|c|c|}
\hline
 & \multicolumn{3}{c|}{$\alpha=1.0$} & \multicolumn{3}{c|}{$\alpha=0.5$} & \multirow{2}{*}{$IC_{50}$$\left(\mu M\right)$}\tabularnewline
\cline{1-7} \cline{2-7} \cline{3-7} \cline{4-7} \cline{5-7} \cline{6-7} \cline{7-7}
Inhibitor & $\varDelta G_{ij}$ (\%) & $\bar{L}$ & $N_{BS}$ & $\varDelta G_{ij}$ (\%) & $\bar{L}$ & $N_{BS}$ & \tabularnewline
\hline
\hline
6M0K & 147.4 & 8.7 & 6 & 70.7 & 9.1 & 7 & $0.04\pm0.002$\tabularnewline
\hline
6LZE & 62.3 & 7.5 & 6 & 13.4 & 8.1 & 6 & $0.053\pm0.005$\tabularnewline
\hline
6Y2G & 57.2 & 7.8 & 6 & -4.0 & 5.8 & 3 & $0.67\pm0.18$\tabularnewline
\hline
\end{tabular}
\par\end{centering}
\caption{Average change individual transmissibility of perturbations between
amino acids in CoV-2 M\protect\textsuperscript{pro} bounded to inhibitors
relative to the free protease. The average path length $\bar{L}$
for paths between the top ten pairs of amino acids according to $\varDelta G_{ij}$
and the number of times a residues in one of these paths is located
in the binding site of the protease, $N_{BS}$.}

\label{relative_changes}
\end{table}

According to the results given in Table \ref{relative_changes} we
can extract the following conclusions. For $\alpha=1.0$, the values
of $\varDelta G_{ij}$ indicate that the three inhibitors increase
the transmissibility of perturbations across the protein in relation to
the free protease. The trend in these percentages of change is parallel
to that of the inhibitory power of the inhibitors. That is, the most
potent inhibitor increases more the transmissibility of effects across
the protease than the second most powerful one, and the least powerful
is the one with the poorer increase in $\varDelta G_{ij}$. However,
neither $\bar{L}$ nor $N_{BS}$ display a consistent pattern of change
in relation to the values of $IC_{50}$$\left(\mu M\right)$. In contrast,
when $\alpha=0.5$ we observe some significant and physically sounded
trends for the three parameters studied. First, the most powerful
inhibitor increases by 71\% the transmissibility of perturbations through
the main protease after its binding. It is followed by the second
most powerful inhibitor, which increases modestly the transmissibility
of perturbations by 13\%. However, the weakest inhibitor does not
increases, but decreases, the transmissibility of perturbations across
the protein. Notice that there is one order of magnitude between the potencies of the first two inhibitors (6M0K and 6LZE) and the
third one (6Y2G). In addition, here, the average length of the shortest
paths connecting the pairs of residues with the largest increase in
the transmissibility of effects follow the same trend as the inhibitory
potency. The most potent inhibitor perturbs an average of 9 residues
per perturbation path. The second most powerful inhibitor perturbs
an average of 8 residues per shortest paths, and the weakest inhibitor
perturbs only 6. This is a physically sounded result as the most powerful
inhibitor produces a stronger effect on the protease which is ``felt''
by a larger number of residues in the structure. Finally, it is also
remarkable that the number of residues in, or close to, the binding
site, correlates with the inhibitory power of the inhibitor.
In this case, the most powerful one starts 70\% of the most important
perturbations according to $\varDelta G_{ij}$ at the binding site,
while the weakest initiates only 30\% of these perturbations at the
binding site.

In terms of the geometric distance between the residues in the perturbed
protease we also observe similar characteristics as for the case of
the length of the shortest path. For instance, for $\alpha=1$ the
average geometric separation of amino acids in the 10 most perturbed
pairs is 33.4 \AA  (6M0K), 28.7 \AA (6LZE) and 29.6 \AA (6Y2G). Here again
we observe a clear lack of correlation with the potency of the inhibitors.
However, for $\alpha=0.5$ we have 35.8 \AA (6M0K), 30.5 \AA (6LZE) and
21.3 \AA (6Y2G), in clear agreement with the trend of inhibitory potency
of the inhibitors.

\section{Conclusions}

There are two main conclusions in the current work. The first is that
we have proposed a generalized fractional-order SI model which includes
the classical SI model as a particular case. We have found an upper
bound to the exact solution of this model, which under given initial
conditions depends only on the Mittag-Leffler matrix function of the
adjacency matrix of the graph. The most important characteristic of this fractional-order
SI model is that it allows to account for long-range interactions
between the nodes of a network as well as for different time-windows
on the transmission of perturbations on networks by tuning the fractional
parameter $\alpha$ of the model. Both characteristics are of great
relevance in many different applications of complex systems ranging
from biological to social systems, and in particular for the study
of protein residue networks.

The second main conclusion of this work is that the fractional-order
SI model allowed us to extract very important information about the
interaction of inhibitors with the main protease of the SARS CoV-2.
This structural information consists in the transmission of perturbations
produced by the inhibitors at the binding site of the protease to
very distant amino acids in other domains of the protein. More importantly,
our findings suggest that the length of this transmission seems to
reflect the potency of the inhibitor. That is, the more powerful inhibitors
transmit perturbations to longer distances through the protein. On
the contrary, weaker inhibitors do not propagate such effect beyond
6 edges from the binding site as average. Consequently, these findings
are important for understanding the mechanisms of actions of such
inhibitors on SARS CoV-2 M\textsuperscript{pro} and helping in the
design of more potent drug candidates against this new coronavirus.
Of course, the current approach can be extended and used for the analysis
of other inhibitors in other proteins not only using experimental
data like in here but using computational analysis of such interactions.

\smallskip
\section*{Acknowledgements}

We thank the Editor and the three anonymous referees for useful suggestions that improve the presentation of this work.

The first author has been partly supported by Project MTM2016-77710-P, DGI-FEDER, of the MCYTS,
Project E26-17R, D.G. Arag\'{o}n, and Project for Young Researchers, Fundaci\'{o}n Ibercaja and Universidad de
Zaragoza, Spain.


\bibliographystyle{abbrv}
\bibliography{fractional_SI}

\begin{thebibliography}{10}

\bibitem{PDB}
Protein data bank: the single global archive for 3{D} macromolecular structure
  data.
\newblock {\em Nucleic acids research}, 47(D1):D520--D528, 2019.

\bibitem{AA}
L.~Abadias, E.~Alvarez, et~al.
\newblock Uniform stability for fractional {C}auchy problems and applications.
\newblock {\em Topological Methods in Nonlinear Analysis}, 52(2):707--728,
  2018.

\bibitem{ALM}
L.~Abadias, C.~Lizama, P.~J. Miana, et~al.
\newblock Sharp extensions and algebraic properties for solution families of
  vector-valued differential equations.
\newblock {\em Banach Journal of Mathematical Analysis}, 10(1):169--208, 2016.

\bibitem{Fractional_compartments}
C.~N. Angstmann, A.~M. Erickson, B.~I. Henry, A.~V. McGann, J.~M. Murray, and
  J.~A. Nichols.
\newblock Fractional order compartment models.
\newblock {\em SIAM Journal on Applied Mathematics}, 77(2):430--446, 2017.

\bibitem{B}
E.~Bazhlekova.
\newblock The abstract {C}auchy problem for the fractional evolution equation.
\newblock {\em Fract. Calc. Appl. Anal}, 1(3):255--270, 1998.

\bibitem{Berry}
H.~Berry.
\newblock Nonequilibrium phase transition in a self-activated biological
  network.
\newblock {\em Physical review E}, 67(3):031907, 2003.

\bibitem{Cooper_Dryden}
A.~Cooper and D.~Dryden.
\newblock Allostery without conformational change.
\newblock {\em European Biophysics Journal}, 11(2):103--109, 1984.

\bibitem{Inhibitors_11}
W.~Dai, B.~Zhang, H.~Su, J.~Li, Y.~Zhao, X.~Xie, Z.~Jin, F.~Liu, C.~Li, Y.~Li,
  et~al.
\newblock Structure-based design of antiviral drug candidates targeting the
  sars-cov-2 main protease.
\newblock {\em Science}, 2020.

\bibitem{diethelm2010analysis}
K.~Diethelm.
\newblock {\em The analysis of fractional differential equations: An
  application-oriented exposition using differential operators of Caputo type}.
\newblock Springer Science \& Business Media, 2010.

\bibitem{Communication_1}
K.~H. DuBay, J.~P. Bothma, and P.~L. Geissler.
\newblock Long-range intra-protein communication can be transmitted by
  correlated side-chain fluctuations alone.
\newblock {\em PLoS computational biology}, 7(9), 2011.

\bibitem{Networks_1}
E.~Estrada.
\newblock {\em The structure of complex networks: theory and applications}.
\newblock Oxford University Press, 2012.

\bibitem{CHAOS}
E.~Estrada.
\newblock Topological analysis of sars cov-2 main protease.
\newblock {\em Chaos, in press}, 2020.

\bibitem{communicability}
E.~Estrada and N.~Hatano.
\newblock Communicability in complex networks.
\newblock {\em Physical Review E}, 77(3):036111, 2008.

\bibitem{SC}
E.~Estrada and J.~A. Rodr\'{\i}guez-Vel\'azquez.
\newblock Subgraph centrality in complex networks.
\newblock {\em Phys. Rev. E}, 71:056103, May 2005.

\bibitem{ML-4}
D.~Fulger, E.~Scalas, and G.~Germano.
\newblock Monte {C}arlo simulation of uncoupled continuous-time random walks
  yielding a stochastic solution of the space-time fractional diffusion
  equation.
\newblock {\em Physical Review E}, 77(2):021122, 2008.

\bibitem{Garrappa_2}
R.~Garrappa.
\newblock Numerical evaluation of two and three parameter mittag-leffler
  functions.
\newblock {\em SIAM Journal on Numerical Analysis}, 53(3):1350--1369, 2015.

\bibitem{Garrappa_1}
R.~Garrappa and M.~Popolizio.
\newblock Computing the matrix mittag-leffler function with applications to
  fractional calculus.
\newblock {\em Journal of Scientific Computing}, 77(1):129--153, 2018.

\bibitem{gorenflo2014mittag}
R.~Gorenflo, A.~A. Kilbas, F.~Mainardi, S.~V. Rogosin, et~al.
\newblock {\em Mittag-Leffler functions, related topics and applications},
  volume~2.
\newblock Springer, 2014.

\bibitem{gorenflo2001fractional}
R.~Gorenflo, F.~Mainardi, E.~Scalas, and M.~Raberto.
\newblock Fractional calculus and continuous-time finance iii: the diffusion
  limit.
\newblock In {\em Mathematical finance}, pages 171--180. Springer, 2001.

\bibitem{PRN_2}
G.~Hu, J.~Zhou, W.~Yan, J.~Chen, and B.~Shen.
\newblock The topology and dynamics of protein complexes: insights from
  intra--molecular network theory.
\newblock {\em Current Protein and Peptide Science}, 14(2):121--132, 2013.

\bibitem{Networks_2}
V.~Latora, V.~Nicosia, and G.~Russo.
\newblock {\em Complex networks: principles, methods and applications}.
\newblock Cambridge University Press, 2017.

\bibitem{Communication_2}
A.~L. Lee et~al.
\newblock Frameworks for understanding long-range intra-protein communication.
\newblock {\em Current Protein and Peptide Science}, 10(2):116--127, 2009.

\bibitem{Lee_et_al}
C.-H. Lee, S.~Tenneti, and D.~Y. Eun.
\newblock Transient dynamics of epidemic spreading and its mitigation on large
  networks.
\newblock In {\em Proceedings of the Twentieth ACM International Symposium on
  Mobile Ad Hoc Networking and Computing}, pages 191--200, 2019.

\bibitem{Drug_design}
S.~Lu, M.~Ji, D.~Ni, and J.~Zhang.
\newblock Discovery of hidden allosteric sites as novel targets for allosteric
  drug design.
\newblock {\em Drug discovery today}, 23(2):359--365, 2018.

\bibitem{Caputo_1}
F.~Mainardi.
\newblock {\em Fractional calculus and waves in linear viscoelasticity: an
  introduction to mathematical models.}
\newblock World Scientific, 2010.

\bibitem{mainardi2000mittag}
F.~Mainardi and R.~Gorenflo.
\newblock On mittag-leffler-type functions in fractional evolution processes.
\newblock {\em Journal of Computational and Applied Mathematics},
  118(1-2):283--299, 2000.

\bibitem{ML-1}
I.~Matychyn.
\newblock On computation of matrix {M}ittag-{L}effler function.
\newblock {\em arXiv preprint arXiv:1706.01538}, 2017.

\bibitem{Bullo}
W.~Mei, S.~Mohagheghi, S.~Zampieri, and F.~Bullo.
\newblock On the dynamics of deterministic epidemic propagation over networks.
\newblock {\em Annual Reviews in Control}, 44:116--128, 2017.

\bibitem{Miotto}
M.~Miotto, L.~Di~Rienzo, P.~Corsi, D.~Raimondo, and E.~Milanetti.
\newblock Simulated epidemics in 3d protein structures to detect functional
  properties.
\newblock {\em arXiv preprint arXiv:1906.05390}, 2019.

\bibitem{Logarithmic}
S.~K. Mishra, M.~Gupta, and D.~K. Upadhyay.
\newblock Fractional derivative of logarithmic function and its applications as
  multipurpose asp circuit.
\newblock {\em Analog Integrated Circuits and Signal Processing},
  100(2):377--387, 2019.

\bibitem{mugnolo}
D.~Mugnolo.
\newblock Dynamical systems associated with adjacency matrices.
\newblock {\em arXiv preprint arXiv:1702.05253}, 2017.

\bibitem{Ottemann}
K.~M. Ottemann, W.~Xiao, Y.-K. Shin, and D.~E. Koshland.
\newblock A piston model for transmembrane signaling of the aspartate receptor.
\newblock {\em Science}, 285(5434):1751--1754, 1999.

\bibitem{paris2002exponential}
R.~Paris.
\newblock Exponential asymptotics of the mittag--leffler function.
\newblock {\em Proceedings of the Royal Society of London. Series A:
  Mathematical, Physical and Engineering Sciences}, 458(2028):3041--3052, 2002.

\bibitem{podlubny1998fractional}
I.~Podlubny.
\newblock {\em Fractional differential equations: an introduction to fractional
  derivatives, fractional differential equations, to methods of their solution
  and some of their applications}.
\newblock Elsevier, 1998.

\bibitem{ML-3}
A.~Sadeghi and J.~R. Cardoso.
\newblock Some notes on properties of the matrix {M}ittag-{L}effler function.
\newblock {\em Applied Mathematics and Computation}, 338:733--738, 2018.

\bibitem{SARS-CoV-2_1}
C.~S.~G. The~International et~al.
\newblock The species severe acute respiratory syndrome-related coronavirus:
  classifying 2019-ncov and naming it {SARS}-{C}o{V}-2.
\newblock {\em Nature Microbiology}, page~1, 2020.

\bibitem{SARS-CoV-2}
F.~Wu, S.~Zhao, B.~Yu, Y.-M. Chen, et~al.
\newblock A new coronavirus associated with human respiratory disease in
  {C}hina.
\newblock {\em Nature}, 579(7798):265--269, 2020.

\bibitem{NMR}
Y.~Xu and M.~Havenith.
\newblock Perspective: Watching low-frequency vibrations of water in
  biomolecular recognition by thz spectroscopy.
\newblock {\em The Journal of chemical physics}, 143(17):170901, 2015.

\bibitem{Crystal-Mpro}
L.~Zhang, D.~Lin, X.~Sun, U.~Curth, C.~Drosten, L.~Sauerhering, S.~Becker,
  K.~Rox, and R.~Hilgenfeld.
\newblock Crystal structure of {SARS}-{C}o{V}-2 main protease provides a basis
  for design of improved $\alpha$-ketoamide inhibitors.
\newblock {\em Science}, 2020.

\bibitem{COVID-19}
P.~Zhou, X.-L. Yang, X.-G. Wang, B.~Hu, et~al.
\newblock A pneumonia outbreak associated with a new coronavirus of probable
  bat origin.
\newblock {\em Nature}, 579(7798):270--273, 2020.

\end{thebibliography}


\begin{thebibliography}{99}
 \normalsize


\bibitem{Networks_1}E. Estrada. The structure of complex networks:
theory and applications. \textit{Oxford University Press}, (2012).

\bibitem{Networks_2}Latora, V. \& Nicosia, V. \& Russo, G. Complex
networks: principles, methods and applications. Cambridge University
Press, (2017).

\bibitem{PRN_2}Hu G, Zhou J, Yan W, Chen J, Shen B. The topology
and dynamics of protein complexes: insights from intra--molecular
network theory. Current Protein and Peptide Science. 2013 Mar 1;14(2):121-32.

\bibitem{Communication_1}DuBay KH, Bothma JP, Geissler PL. Long-range
intra-protein communication can be transmitted by correlated side-chain
fluctuations alone. PLoS computational biology. 2011 Sep;7(9).

\bibitem{Communication_2}Lee AL. Frameworks for understanding long-range
intra-protein communication. Current Protein and Peptide Science.
2009 Apr 1;10(2):116-27.

\bibitem{Ottemann}Ottemann KM, Xiao W, Shin YK, Koshland DE (1999)
A piston model for transmembrane signaling of the aspartate receptor.
Science 285: 1751--1754.

\bibitem{Cooper_Dryden}Cooper A, Dryden DT (1984) Allostery without
conformational change. a plausible model. Eur Biophys J 11: 103--109.

\bibitem{Drug design}Lu S, Ji M, Ni D, Zhang J. Discovery of hidden
allosteric sites as novel targets for allosteric drug design. Drug
discovery today. 2018 Feb 1;23(2):359-65.

\bibitem{Berry}Berry H. Nonequilibrium phase transition in a self-activated
biological network. Physical review E. 2003 Mar 14;67(3):031907.

\bibitem{Miotto}Miotto M, Di Rienzo L, Corsi P, Raimondo D, Milanetti
E. Simulated Epidemics in 3D Protein Structures to Detect Functional
Properties. J. Chem. Inf. Model. 2020, 60, 1884- 1891.

\bibitem{Lee_et_al}Lee CH, Tenneti S, Eun DY. Transient Dynamics
of Epidemic Spreading and Its Mitigation on Large Networks. InProceedings
of the Twentieth ACM International Symposium on Mobile Ad Hoc Networking
and Computing 2019 Jul 2 (pp. 191-200).

\bibitem{SARS-CoV-2_1}Gorbalenya, A., Baker, S. \& Baric, R. Coronaviridae
Study Group of the International Committee on Taxonomy of Viruses:
The species Severe acute respiratory syndrome-related coronavirus:
classifying 2019-nCoV and naming it SARS-CoV-2. \textit{Nature Microbiol}.
(2020) 3(04).

\bibitem{mugnolo}D. Mugnolo. Dynamical systems associated with adjacency matrices.
\textit{Discrete $\&$ Continuous Dynamical Systems}- B, 2018, \textbf{23} (5) : 1945-1973.


\bibitem{SARS-CoV-2}Wu, F., Zhao, S., Yu, B., Chen, Y. M., Wang,
W., Song, Z. G., Hu, Y., Tao, Z. W., Tian, J. H., Pei, Y. Y. \& Yuan,
M. L. A new coronavirus associated with human respiratory disease
in China. \textit{Nature} \textbf{579}, 265-269 (2020).

\bibitem{COVID-19}Zhou, P., Yang, X. L., Wang, X. G., Hu, B., Zhang,
L., Zhang, W., Si, H. R., Zhu, Y., Li, B., Huang, C. L. \& Chen, H.
D. A pneumonia outbreak associated with a new coronavirus of probable
bat origin. \textit{Nature} \textbf{579}, 270-273 (2020).

\bibitem{Crystal-Mpro}Zhang, L., Lin, D., Sun, X., Curth, U., Drosten,
C., Sauerhering, L., Becker, S., Rox, K. \& Hilgenfeld, R. Crystal
structure of SARS-CoV-2 main protease provides a basis for design
of improved \textgreek{a}-ketoamide inhibitors. \textit{Science} (2020)
Mar 20.

\bibitem{PDB}Protein Data Bank: the single global archive for 3D
macromolecular structure data. \textit{Nucleic Acids Res}. \textbf{47},
D520-D528 (2019).

\bibitem{Inhibitor_2}Jin Z, Du X, Xu Y, Deng Y, Liu M, Zhao Y, Zhang
B, Li X, Zhang L, Peng C, Duan Y. Structure of Mpro from COVID-19
virus and discovery of its inhibitors. bioRxiv. 2020 Jan 1.

\bibitem{ernestoCovid}E. Estrada. Topological Analysis of SARS CoV-2 Main Protease.
bioRxiv. 2020 Jan 1. https://doi.org/10.1101/2020.04.03.023887.

\bibitem{Inhibitor_1}Mesecar, A. D. A taxonomically-Driven approach
to development of pot broad-Spectrum inhibitors of coronavirus main
proteas including sars-Cov-2 (covid-19). To be published.

\bibitem{Bullo}Mei W, Mohagheghi S, Zampieri S, Bullo F. On the dynamics
of deterministic epidemic propagation over networks. Annual Reviews
in Control. 2017 Jan 1;44:116-28.

\bibitem{SC}Estrada E, Rodriguez-Velazquez JA. Subgraph centrality
in complex networks. Physical Review E. 2005 May 6;71(5):056103.

\bibitem{communicability}Estrada E, Hatano N. Communicability in
complex networks. Physical Review E. 2008 Mar 11;77(3):036111.

\bibitem{Caputo_1}J. Klafter, S.-C. Lim, and R. Metzler, Editors
Fractional dynamics in physics: recent advances, World Scientific, Singapore, 2011.

\bibitem{Caputo_2}Delbourgo D. A Dirichlet series expansion for the
p-adic zeta-function. Journal of the Australian Mathematical Society.
2006 Oct;81(2):215-24.

\bibitem{Logarithmic}Mishra SK, Gupta M, Upadhyay DK. Fractional
derivative of logarithmic function and its applications as multipurpose
ASP circuit. Analog Integrated Circuits and Signal Processing. 2019
Aug 15;100(2):377-87.

\bibitem{Fractional compartments}Angstmann CN, Erickson AM, Henry
BI, McGann AV, Murray JM, Nichols JA. Fractional order compartment
models. SIAM Journal on Applied Mathematics. 2017;77(2):430-46.

\bibitem{AA} L. Abadias, E. \'{A}lvarez, Uniform stability for fractional
Cauchy problems and applications. Topol. Methods Nonlinear Anal. 52
(2), 707-728 (2018).

\bibitem{ALM} L. Abadias, C. Lizama and P. J. Miana, Sharp extensions
and algebraic properties for solutions families of vector-valued differential
equations. Banach J. Math. Anal. \textbf{10} (2016), 169--208.

\bibitem{B} E. Bajlekova, The abstract Cauchy problem for the fractional
evolution equation. Fract. Calc. Appl. Anal. \textbf{1} (3) (1998),
255--270.

\bibitem{Matrix functions}Benzi M, Boito P. Matrix functions in network
analysis. GAMM Mitteilungen. 2020.

\bibitem{ML-1}Matychyn I. On Computation of Matrix Mittag-Leffler
Function. arXiv preprint arXiv:1706.01538. 2017 Jun 5.

\bibitem{ML-2}Garrappa R, Popolizio M. Computing the matrix Mittag-Leffler
function with applications to fractional calculus. Journal of Scientific
Computing. 2018 Oct 1:1-25.

\bibitem{ML-3}Sadeghi A, Cardoso JR. Some notes on properties of
the matrix Mittag-Leffler function. Applied Mathematics and Computation.
2018 Dec 1;338:733-8.

\bibitem{ML-4}Fulger D, Scalas E, Germano G. Monte Carlo simulation
of uncoupled continuous-time random walks yielding a stochastic solution
of the space-time fractional diffusion equation. Physical Review E.
2008 Feb 25;77(2):021122.

\bibitem{Silver}Estrada E, Silver G. Accounting for the role of long
walks on networks via a new matrix function. Journal of Mathematical
Analysis and Applications. 2017 May 15;449(2):1581-600.

\end{thebibliography}

 \bigskip \smallskip

 \it

 \noindent
\bigskip \smallskip

 \it

 \noindent
\bigskip \smallskip

 \it

 \noindent
$^{1}$Departamento de Matem\'{a}ticas, Facultad de Ciencias\\ Universidad de Zaragoza, 50009 Zaragoza, Spain.

e-mail: labadias@unizar.es\\

$^{2}$Instituto Universitario
de Matem\'{a}ticas y Aplicaciones\\ Universidad de Zaragoza, 50009 Zaragoza,
Spain.

email: estrada66@unizar.es\\

$^{3}$Laboratoire Jacques-Louis Lions, Universit\'{e} Pierre-et-Marie-Curie
(UPMC)\\ 4 place Jussieu, 75005, Paris, France.

email: estradarodriguez@ljll.math.upmc.fr  \\

$^{4}$ARAID Foundation,
Government of Arag\'{o}n\\ 50018 Zaragoza, Spain.

email: estrada66@posta.unizar.es\\

\end{document}